\documentclass[submission,copyright,creativecommons]{eptcs}
\providecommand{\event}{GANDALF 2025} % Name of the event you are submitting to

\usepackage{iftex}

\ifpdf
  \usepackage{underscore}         % Only needed if you use pdflatex.
  \usepackage[T1]{fontenc}        % Recommended with pdflatex
\else
  \usepackage{breakurl}           % Not needed if you use pdflatex only.
\fi

\usepackage[pdftex]{color}
\usepackage{amsthm}
\usepackage{amssymb}
\usepackage{amsmath}
\usepackage{txfonts}
\usepackage{float}
\usepackage{subcaption}
\usepackage{url}
\usepackage{times}
\usepackage{stmaryrd}
\usepackage{todonotes}
\usepackage{booktabs}
\usepackage{latexsym}
\usepackage{wasysym}
\usepackage{xspace}
\usepackage[export]{adjustbox}

\usepackage{url}
\usepackage{stmaryrd}

\usepackage{algorithm}
\usepackage{algorithmic}

\theoremstyle{plain}% default
\newtheorem{thm}{thm}[section]
\newtheorem{theorem}[thm]{Theorem}
\newtheorem{lem}[thm]{Lemma}
\newtheorem{prop}[thm]{Proposition}

\newtheorem{defn}{Definition}[section]

\newtheorem*{rem}{Remark}

\newtheorem{prob}{Problem}

\newcommand{\be}{\begin{enumerate}}
\newcommand{\ee}{\end{enumerate}}
\newcommand{\bi}{\begin{itemize}}
\newcommand{\ei}{\end{itemize}}

\newcommand\etal{\textit{et al.}}
\newcommand\ie{\textit{i.e.}}

\newcommand{\ra}{\rightarrow}

\newcommand{\Lra}{\Leftrightarrow}

\newcommand{\cF}{\mathcal{F}}
\newcommand{\cG}{\mathcal{G}}
\newcommand{\cH}{\mathcal{H}}
\newcommand{\cJ}{\mathcal{J}}

\newcommand{\N}{\mathbb{N}}

\newcommand{\B}{\mathbb{B}}

\newcommand{\true}{\mathsf{true}}
\newcommand{\false}{\mathsf{false}}

\usepackage{tikz}

\usetikzlibrary{arrows,arrows.meta,automata,positioning}
\usetikzlibrary{cd} % for commuting diagrams
\usetikzlibrary{tikzmark,shapes,fit,calc}
\tikzstyle{smallblock} = [draw, fill=white, square state, 
   minimum height=2cm, minimum width=3cm]
\tikzstyle{block} = [draw, fill=white, square state, 
   minimum height=2cm, minimum width=4cm]
\tikzstyle{bigblock} = [draw, fill=white, square state, 
   minimum height=4cm, minimum width=8cm]
\tikzstyle{input} = [coordinate]
\tikzstyle{output} = [coordinate]
\tikzstyle{pinstyle} = [pin edge={to-,thin,black}]

\tikzset{
->, % makes the edges directed
>=stealth', % makes the arrow heads bold
node distance=1cm, % specifies the minimum distance between two nodes. Change if necessary.
initial text=$ $, % sets the text that appears on the start arrow
semithick
}

\tikzset{every state/.style={minimum size=15pt}}
\tikzset{
  every picture/.append style={remember picture,inner xsep=0,inner ysep=0},
}

\tikzset{every loop/.style={min distance=1mm,in=-30,out=30,looseness=5}}
\tikzcdset{arrow style=tikz, diagrams={>=stealth}}

\tikzset{square state/.style={draw,rectangle,minimum size=5mm}}

\usepackage{forest}

\newcommand{\occ}{\mathit{occ}}
\newcommand{\infi}{\mathit{inf}}
\newcommand{\reach}{\text{\textsc{Reach}}}
\newcommand{\safe}{\text{\textsc{Safe}}}
\newcommand{\buchi}{\text{\textsc{B{\"u}chi}}}
\newcommand{\cobuchi}{\text{\textsc{coB{\"u}chi}}}
\newcommand{\parity}{\text{\textsc{Parity}}}
\newcommand{\muller}{\text{\textsc{Muller}}}

\newcommand{\p}{\mathsf{P}}
\newcommand{\np}{\mathsf{NP}}

\newcommand{\pspace}{\mathsf{PSPACE}}

\newcommand{\sat}{\mathsf{SAT}}
\newcommand{\osat}{\oplus\mathsf{SAT}}

\newcommand{\tqbf}{\mathsf{TQBF}}

\newcommand{\swdp}{\mathsf{SWDP}}
\newcommand{\podp}{\mathsf{PODP}}

\newcommand{\state}{\mathsf{St}}
\newcommand{\agt}{\mathsf{Agt}}
\newcommand{\act}{\mathsf{Act}}
\newcommand{\tr}{\mathsf{tr}}
\newcommand{\avb}{\mathsf{avb}}
\newcommand{\plays}{\mathsf{Plays}}
\newcommand{\last}{\mathsf{last}}
\newcommand{\hist}{\mathsf{Hist}}
\newcommand{\out}{\mathsf{Out}}
\newcommand{\obj}{\mathsf{Obj}}
\newcommand{\pay}{\mathsf{Payoff}}

\newcommand{\win}{\mathsf{Win}}
\newcommand{\sw}{\mathsf{sw}}

\newcommand{\susp}{\mathsf{Susp}}

\newcommand{\eve}{\textsf{Eve}}
\newcommand{\adam}{\textsf{Adam}}

\newcommand{\proj}{\mathit{proj}}

\newcommand{\ltl}{\mathsf{LTL}}
\newcommand{\gr}{\mathsf{GR(1)}}

\title{The Complexity of Pure Strategy Relevant Equilibria in Concurrent Games}
\author{Purandar Bhaduri
\institute{Indian Institute of Technology Guwahati, India}
\email{pbhaduri@iitg.ac.in}}

\def\titlerunning{Complexity of Pure Strategy Relevant Equilibria}
\def\authorrunning{P. Bhaduri}

\begin{document}
\maketitle

\begin{abstract}
We study rational synthesis problems for concurrent games with $\omega$-regular objectives. Our model of rationality considers only pure strategy Nash equilibria that satisfy either a social welfare or Pareto optimality condition with respect to an $\omega$-regular objective for each agent. This extends earlier work on equilibria in concurrent games, without consideration about their quality. Our results show that the existence of Nash equilibria satisfying social welfare conditions can be computed as efficiently as the constrained Nash equilibrium existence problem. On the other hand, the existence of Nash equilibria satisfying the Pareto optimality condition possibly involves a higher upper bound, except in the case of B{\"u}chi and Muller games, for which all three problems are in
the classes $\p$ and $\pspace$-complete, respectively.
\end{abstract}

\thanks{This work has been supported by the MATRICS grant MTR/2021/000003 from SERB India, titled "Games and Controller Synthesis".}

\section{Introduction}
\label{sec:intro}

Infinite games on finite graphs play a fundamental role in the automated synthesis of reactive systems from their specification~\cite{finkbeiner2016synthesis,BloemCJ18}. The goal of reactive synthesis is to design a system that meets its specification in all possible environments. This problem can be modelled as a zero-sum game between the system and the environment, where a winning strategy for the system yields the design of a correct-by-construction controller.

In many situations involving autonomous agents, such as robots, drones, and
autonomous vehicles, a purely adversarial view of the environment of a system is not appropriate. Instead, each agent should be viewed as trying to satisfy her own objective rather than preventing the other agents from meeting theirs. This gives rise to the notion of \emph{rational synthesis}~\cite{fisman2010rational,kupferman2016synthesis}, where we restrict our attention to agent behaviours that arise from game-theoretic equilibria. The most widely studied equilibrium is the Nash equilibrium (NE)~\cite{nash1950equilibrium}, a strategy profile for which no agent has the incentive to unilaterally deviate from her strategy.

In this paper we investigate the complexity of computing certain desirable Nash
equilibria in concurrent games played on finite graphs. We call these equilibria
\emph{relevant} after \cite{brihaye2021relevant}. Our games have multiple
agents (or players), each with an $\omega$-regular objective. The games are concurrent, where the agents make their moves simultaneously. Each agent tries to fulfil her own objective while being subject to the decisions of the other agents. Such games provide a common framework for modelling systems with multiple, distributed, and autonomous agents.

It is well known that for such games pure strategy Nash equilibria may not exist and may not be unique when they do. Some equilibria may not be optimal, where only a few or none of the agents meet their objectives. We focus on two measures of the
quality of equilibria -- social welfare and Pareto optimality.

Intuitively, the social welfare criterion considers the number of agents who
meet their objectives in a given equilibrium. Pareto optimal equilibria, on the
other hand, consider equilibria which are maximal in the sense that no additional
agent can meet her objective in any behaviour of the system, whether it is the
result of an equilibrium or not. 

The specific problems which we are interested in are (i) the Constrained NE Existence Problem, asking whether there is an equilibrium satisfying a given lower and upper threshold for the payoff of each agent, (ii) the Social Welfare Decision Problem ($\swdp$), asking whether there is an equilibrium where the number of agents meeting their objectives is above a given threshold,
and (iii) the Pareto Optimal Decision Problem ($\podp$), where the problem is to decide whether there is an equilibrium such that no other strategy profile, whether an equilibrium or not, results in a strict superset of agents meeting their objectives. The complexity of the Constrained NE Existence Problem for different $\omega$-regular objectives was extensively studied by Bouyer \etal~\cite{bouyer2015pure}. We mention them here as a point of reference and also because many of our results depend on them.

A summary of our results for $\swdp$ and $\podp$ is shown in Table~\ref{tab:results}. The results show that the existence of NE satisfying a social welfare condition can be decided as efficiently as the constrained NE existence problem. However, a Pareto optimality condition seems to entail a higher upper bound, except for B{\"u}chi and Muller objectives, for which all three problems are in the classes $\p$ and $\pspace$-complete, respectively.

\begin{table}[t]
\centering
\begin{tabular}{ |p{2cm}||p{2.75cm}|p{2.2cm}|p{2.5cm}|  }
%\begin{tabular}{|r||c|c|c|} 
\hline
{\bf Objective} & {\bf Constrained NE Existence \cite{bouyer2015pure} } & {\bf Social Welfare} ($\swdp)$ & {\bf Pareto Optimality} ($\podp$) \\
\hline
Reachability &  $\np$-c & $\np$-c & $\np$-h, $\p^\np$  \\
\hline
Safety & $\np$-c  & $\np$-c & $\np$-h, $\p^\np$  \\ 
\hline
B{\"u}chi & $\p$-c & $\p$ & $\p$ \\
\hline
coB{\"u}chi & $\np$-c & $\np$-c & $\np$-h, $\p^\np$   \\
\hline
Parity & $\p^{\np}_\parallel$-c  & $\np$-h, $\p^{\np}_\parallel$  & $\p^\np$ \\
\hline
Muller & $\pspace$-c  & $\pspace$-c & $\pspace$-c \\
\hline
\end{tabular}
\caption{Summary of complexities for relevant equilibria}
\label{tab:results}
\end{table}

\subsection{Related Work}
\label{subsec:related}

Concurrent game structures were introduced by Alur \etal~\cite{alur1997alternating}
for modelling the behaviour of open systems containing both system and environment components, called players or agents. Every state transition in such a game is determined by a choice of move by each player and models their synchronous composition. 

The seminal work related to problems on Nash equilibria for concurrent games is
by Bouyer \etal~\cite{bouyer2015pure}. This paper studies Nash equilibria for
pure-strategy multiplayer concurrent deterministic games for various
$\omega$-regular objectives and their generalisations. It gives
comprehensive results for the Value Problem, the NE Existence Problem and the
Constrained NE Existence Problem for all these objectives. The present work is
an extension where we also study the existence of desirable equilibria,
namely the existence of NE with lower bounds on social welfare and those satisfying
Pareto optimality.

The term `relevant equilibria' was introduced by Brihaye \etal~\cite{brihaye2021relevant}
to refer to desirable equilibria such as those with high social welfare and those that are Pareto optimal. This paper was in the context of turn-based multiplayer games with quantitative reachability objectives. In contrast, our work is about concurrent qualitative games with the $\omega$-regular objectives specified above.

Another line of work that deals with infinite concurrent games on finite graphs
with both qualitative and quantitative objectives is by Gutierrez
\etal~\cite{gutierrez2023complexity}. In contrast to our work, their focus is on
$\ltl$ and $\gr$ objectives, a setting where the input sizes can
be exponentially more succinct and the expressive power somewhat restricted.

Condurache \etal~\cite{condurache2016complexity} have also investigated rational synthesis problems for concurrent games for the same objectives that we consider in this paper. However, their interest lies in \emph{non-cooperative} rational synthesis, where the aim is to synthesise a strategy for agent~$0$ such that with respect to \emph{all} NE involving the rest of the agents, agent~$0$'s payoff is always $1$. Computationally, this is a harder problem than the \emph{cooperative} rational synthesis problem in our work, where we are interested in the existence of \emph{at least one} NE satisfying some constraints. Moreover, \cite{condurache2016complexity} does not consider the quality of NEs, in contrast to the present work.

\section{Preliminaries}
\label{sec:prelims}

\subsection{Notation}
\label{subsec:not}
$\B$ and $\N$ denote the sets $\{0,1\}$ of Boolean values and of
natural numbers, respectively. A \emph{word} over a finite alphabet $\Sigma$ is
a finite sequence of symbols from $\Sigma$. $\Sigma^\ast$ denotes the set of all
words, including the empty word $\varepsilon$. $\Sigma^\omega$ denotes the set
of all \emph{infinite words}, \ie, infinite sequences over $\Sigma$. For $m,n
\in \N$ with $m \leq n$ and a (finite or infinite) word $\alpha$, we denote by
$\alpha[j]$ the $j+1$-th letter of $\alpha$ and by $\alpha[m,n]$ the finite word
$\alpha[m]\alpha[m+1]\ldots\alpha[n]$. 
% The infinite word $\alpha[m]\alpha[m+1]\ldots$ is denoted by $\alpha[m,\omega]$.

\subsection{Concurrent Game Structures}
\label{subsec:cgs}
\begin{defn}

A \emph{concurrent game structure} (CGS) is a tuple $\cG = (\state, \agt,
\act, \avb, \tr)$ where $\state$ is a finite non-empty set of states, $\agt = \{1,\ldots,n\}$ is the set of agents (or players), $\act$ is a finite set of actions, the map $\avb: \state \times \agt \ra 2^\act \setminus \{\emptyset\}$ indicates the actions available in a given state for a given agent, and $\tr : \state \times \act^n \ra \state$ is the transition function.
\end{defn}

An \emph{action profile} or \emph{move} $\bar{a} = (a_1,a_2, \ldots, a_n) \in 
\act^n$ is just an $n$-tuple of actions. Here, $a_i$ is the action taken by
agent~$i$. We often write $\bar{a}(i)$ for $a_i$. We say $\bar{a}$ is
\emph{legal} at a state $s$ if $a_i \in \avb(s,i)$ for all $i \in \agt$.
We call a CGS \emph{turn-based} if, for each state, the set of available moves
is a singleton for all but at most one player; such a player is said to \emph{own}
the state.\footnote{We prefer the term player to agent when referring to 
turn-based games.}

\begin{defn}
A \emph{play} in the CGS $\cG$ is a sequence of states 
$\rho = s_0  s_1  s_2  \ldots$ such that $s_{j+1} =
\tr(s_j,\bar{a}_j)$ for some legal move $\bar{a}_j$ in state
$s_j$, for all $j \geq 0$.
\end{defn}

Let $\plays(\cG)$ be the set of all plays in the CGS $\cG$. A \emph{history} $h$
is any finite prefix of a play. Let $\hist(\cG)$ be the set of
all histories in $\cG$. We often drop $\cG$ when referring to plays and
histories when it is clear from the context. The last element of a history $h$ is denoted by $\last(h)$.

\begin{defn} 
A strategy for agent~$i$ is a map $\sigma_i : \hist \ra \act$ such that $\sigma_i(h) \in \avb(\last(h),i)$. A
\emph{strategy profile} is a tuple $\bar{\sigma} = \langle \sigma_1, \ldots,
\sigma_n \rangle$ of strategies, one for each agent.
\end{defn}

By convention, $\bar{\sigma}_{-i}$ is the tuple of strategies excluding that of
agent~$i$ and $\langle \bar{\sigma}_{-i},\sigma'_i \rangle$ is obtained from the
profile $\bar{\sigma}$ by substituting agent~$i$'s strategy $\sigma_i$
by $\sigma'_i$. Note that our strategies are \emph{pure}, \ie, they do not involve any randomisation.

We say that a play $\rho$ is \emph{compatible} with a strategy $\sigma_i$ of agent~$i$ if for every prefix
$\rho[0,k]$ of $\rho$ with $k \geq 0$ we have $\sigma_i(\rho[0,k]) =
a_{i_k}$ and $\tr(\rho[k],\langle a_{-i_k}, a_{i_k} \rangle) = \rho[k+1]$, for some action profile $a_{-i_k}$ of the other agents making $\langle a_{-i_k}, a_{i_k} \rangle$ legal in $\rho[k]$. We can define compatibility between a history and an agent~$i$ strategy in a similar way. For a coalition $P \subseteq \agt$, and a tuple $\sigma_P$ of strategies for the agents in $P$, we write $\out_\cG(\sigma_P)$ for the set of plays (called \emph{outcomes}) in $\cG$ that are compatible with strategy $\sigma_i$ for every $i \in P$. Note that a strategy profile $\bar{\sigma}$ and an initial state $s$ uniquely define a play $\out(s,\bar{\sigma})$, referred to as its \emph{outcome}. 

% Let $\sigma_i$ be a strategy for agent~$i$. We denote by $\plays(\sigma_i)$ and
% $\hist(\sigma_i)$ the set of all plays and histories that are compatible with
% $\sigma_i$, respectively.

\begin{rem}
We assume that the transition function $\tr$ is represented explicitly as a table when $\cG$ is an input to an algorithm. Its size $|\tr|$ is $\sum_{s \in \state}\prod_{i \in \agt} \avb(s,i)\cdot\lceil \log(|\state|)\rceil$ and this can be exponential in the number $n$ of agents.
\end{rem}

\subsection{Concurrent Games and Solution Concepts}
\label{subsec:soln}

\paragraph{\bf Omega-regular Games}

A concurrent (or multiplayer) game is a pair $\langle \cG, (\obj_i )_{i \in
\agt} \rangle$ where $\cG$ is a CGS and $\obj_i \subseteq \state^\omega$ is the
\emph{objective} for agent~$i$. Thus, an objective is a set of infinite sequences  
of states in $\cG$. For us, an objective can be any one of safety,
reachability, B{\"u}chi, coB{\"u}chi, parity or Muller, defined below.
For a play $\rho$ in a CGS $\cG$ we write $\occ(\rho)$ for the states that occur
in $\rho$ and $\infi(\rho)$ for the states that occur infinitely often
in $\rho$, \ie, $\occ(\rho) = \{s \in \state \mid \exists j \geq 0.\; s =
\rho[j]\}$ and $\infi(\rho) = \{s \in \state \mid \forall j \geq 0.\; \exists k \geq
j.\; s = \rho[k]\}$. Then, we consider the following objectives:

\be
\item\emph{Reachability:\/} Given a set $F \subseteq \state$ of \emph{target states}, $\reach(F) = \{\rho \in \state^\omega \mid \occ(\rho) \cap F \neq \emptyset\}$;
\item\emph{Safety:\/} Given a set $F \subseteq \state$ of \emph{unsafe states}, $\safe(F) = \{\rho \in \state^\omega \mid \occ(\rho) \cap F = \emptyset\}$;
\item\emph{B{\"u}chi:\/} Given a set $F \subseteq \state$ of \emph{accept states}, $\buchi(F) =  \{\rho \in \state^\omega \mid \infi(\rho) \cap F \neq \emptyset\}$;
\item\emph{coB{\"u}chi:\/} Given a set $F \subseteq \state$ of \emph{reject states}, $\cobuchi(F) =  \{\rho \in \state^\omega \mid \infi(\rho) \cap F = \emptyset\}$;
\item\emph{Parity:\/} For a given \emph{priority function} $p:\state\ra\N$, $\parity(p) = \{\rho \in \state^\omega \mid
\min\{p(s) \mid s \in \infi(\rho)\} \text{ is even}\}$;
\item\emph{Muller:\/} For a given finite set $C$ of \emph{colours}, a {colouring function} $c: \state \ra C$ and a set $\cF \subseteq 2^C$, $\muller(\varphi) = \{\rho \in \state^\omega \mid \infi(c(\rho)) \in \cF \}$. Here $c(\rho)$ is the infinite sequence of colours of the states in the sequence $\rho$ and $\infi(c(\rho))$ is the set of colours appearing infinitely often in the sequence $c(\rho)$.
\ee

For a given play $\rho$, the \emph{payoff} of $\rho$, denoted $\pay(\rho)$, is given by the tuple $\langle \pay_i(\rho) \rangle_{i \in \agt}$. Here, $\pay_i(\rho)  \in \{0,1\}$, the payoff of agent~$i$, is defined by $\pay_i(\rho) =1 \Lra \rho \in \obj_i$. We say agent~$i$ wins the play $\rho$ if her payoff is $1$ and she loses $\rho$ otherwise. For a given state $s$ in $\cG$, we write $\pay(s,\bar{\sigma})$ (respectively, $\pay_i(s,\bar{\sigma})$) for the payoff (respectively, payoff of agent~$i$) for the unique play $\rho$ that is the outcome of the strategy profile $\bar{\sigma}$ starting from state $s$.

A two-player concurrent game is \emph{zero-sum} if for any play $\rho$, Player~$1$ wins $\rho$ if, and only if, Player~$2$ loses it. This is a purely adversarial setting, where one player loses if the other wins and there are no ties. In a zero-sum game, we are interested in finding a winning strategy for a player from a given state if it exists. Such a strategy allows the player to win no matter how the other player moves. In non-zero-sum games, where each player has her own objective and is not necessarily trying to play spoilsport, winning strategies are too restrictive and seldom exist. Instead, the notion of an equilibrium, a strategy profile that is best possible for each player
given the strategies of the other players in the profile, is the key concept. The most celebrated equilibrium in the literature is that of Nash equilibrium, defined below.

%\paragraph{\bf Quantitative Games: Mean-Payoff Games}
%
%For an infinite sequence $\alpha = a_0a_1 \ldots$ of reals, the
%\emph{mean-payoff} value of $\alpha$ is defined as \[\mpay(\alpha) =\lim
%\inf_{n \ra \infty} \frac{1}{n} \sum_{j = 0}^{n-1} a_j.\] A \emph{mean-payoff}
%game is a tuple $G= \langle \cG, (w_i)_{i \in \agt}\rangle$, where $\cG$ is a
%CGS and $w_i : \state \ra \Z$ is a weight function mapping every state to an
%integer for each agent~$i$. In a mean-payoff game, a play $\rho$ induces a
%sequence of weights $\nu_i(\rho) = w_i(s_0)w_i(s_1) \ldots$ for each agent~$i$.
%The payoff of agent~$i$ in $\rho$ is then defined as $\pay_i(\rho) = \mpay(\nu_i(\rho))$.

\paragraph{\bf Nash Equilibrium and Relevant Equilibria}

The solution concept for games we consider in this paper is the Nash equilibrium~\cite{nash1950equilibrium}, a strategy profile in which no agent has the incentive to unilaterally change her strategy.

\begin{defn}
Let $\cG$ be a concurrent game and let $s$ be a state of $\cG$. A strategy profile $\bar{\sigma}$ is a \emph{Nash equilibrium} (NE) of $\cG$ from $s$ if for every agent~$i$ and every strategy $\sigma_i'$ of agent~$i$, it is the case that 
$\pay_i(s,\bar{\sigma}) \geq \pay_i(s,\langle\bar{\sigma}_{-i},\sigma_i' \rangle)$.
\end{defn}

%% Significance of lassos
In the kind of games we consider, the payoffs of agents for a given play $\rho$ only depend on the set of states that are visited and the set of states that are visited infinitely often in $\rho$. For such games, the outcomes of Nash equilibria can be taken to be ultimately periodic sequences as shown in \cite{bouyer2015pure} (Proposition 3.1). We restate the result here.

\begin{prop}
\label{prop:lasso}
Suppose $\cG$ is a concurrent game where for any pair $\rho,\rho'$ of plays, $\occ(\rho) = \occ(\rho')$ and $\inf(\rho) = \inf(\rho')$ imply $\pay(\rho)=\pay(\rho')$. If $\rho$ is an outcome of an NE in $\cG$ then there is an NE with outcome $\rho'$ of the form $\alpha_1.\alpha_2^\omega$ such that $\pay(\rho)=\pay(\rho')$, where the lengths $|\alpha_1|$ and $|\alpha_2|$ have an upper bound of $|\state|^2$.
\end{prop}

In concurrent games, pure strategy Nash equilibria may not exist, or multiple
equilibria may exist. We identify the equilibria that are desirable and call
them \emph{relevant} after \cite{brihaye2021relevant}. For example, an equilibrium in which the number of players who meet their objectives is maximum among all equilibria (\ie, one that maximises the \emph{social welfare} defined below) would be considered a relevant one. 

The relevant equilibria we consider are based on social welfare and Pareto
optimality. Given a play $\rho$ in a concurrent
game $\langle \cG, (\obj)_i \rangle$, we denote by $\win(\rho)$ the set of
agents who meet their objective along $\rho$, \ie, $\win(\rho) = \{ i \in \agt
\mid \rho \in \obj_i\}$. The \emph{social welfare} $\sw(\rho)$ of $\rho$ is
$|\win(\rho)|$, \ie, the number of agents who meet their objectives in
$\rho$. For a given state $s$ in $\cG$ and a strategy profile $\bar{\sigma}$,
we write $\sw(s,\bar{\sigma})$ for $\sw(\rho)$ where $\rho = \out(s,\bar{\sigma})$.

To define Pareto optimality in a game, let $P = \{\langle w_i(\rho) \rangle_{i \in \agt} \mid \rho \in \plays(\cG)\}$, where $w_i(\rho)=1$ if agent~$i$ wins 
$\rho$ and $0$ otherwise, \ie, the set $P$ is the set of winner profiles for all plays in $\cG$. Then a winner profile $p \in P$ is \emph{Pareto optimal} if it is maximal in $P$ with respect to the componentwise ordering $\leq_P$ on $P$ where $0 < 1$ in each component. This means no other agent can be added to the set of winners in $p$, \ie, $p$ represents a maximal set of winners along any play.

%For mean-payoff games, social welfare and Pareto optimality are classical. A
%\emph{social welfare function} $\sw : \R^n \ra \R$ maps an $n$-tuple of reals to
%a real representing some aggregated notion of payoff. For a strategy profile
%$\sigma$, the social welfare of $\sigma$ is given by
%$\sw(\pay_1(\sigma),\ldots,\pay_n(\sigma))$, and is denoted by $\sw(\sigma)$ for
%brevity. We focus on the following social welfare functions, called utilitarian
%and egalitarian, as in \cite{gutierrez2023complexity}. The \emph{utilitarian} social
%welfare function is defined as $\usw(\sigma) =\sum_{i \in \agt}\pay_i(\sigma)$,
%\ie, the sum of the payoffs of the agents for the given profile. The
%\emph{egalitarian} social welfare function is defined by $\esw(\sigma) = \min_{i
%\in \agt}\{\pay_i(\sigma)\}$, \ie, the minimum payoff guaranteed for any player.
%
%Pareto optimality for mean-payoff games is defined as follows. A payoff profile
%$p$ in $P = \{(\pay_i(\rho))_{i \in \agt} \mid \rho \in \plays(cG)\}$ is
%\emph{Pareto optimal} if it is maximal in $P$ with respect to the componentwise
%order $\leq_P$ on $P$ using the usual order on reals for each component.
%Note that for qualitative games, the Pareto optimality  based on only the utilitarian
%notion makes sense and is identical to what we have defined.

\subsection{Rational Synthesis Problems}
\label{subsec:probs}

The rational synthesis problem generalises the synthesis problem to multiagent
systems. The aim is to synthesise a game-theoretic equilibrium, a
Nash equilibrium in our case, that satisfies additional desirable properties.
The problems we study are detailed below.

%\begin{prob}[\bf Threshold Decision Problem ($\tdp$)] Given a game and a
%threshold tuple $v \in (\Q \cup \{-\infty\})^n$, decide whether \emph{there
%exists} an NE $\sigma$ such that $(\pay_i(\sigma))_{i \in \agt} \geq v$.
%\end{prob}
%
%By imposing lower bounds on the payoff profile in a qualitative game, we look
%for NE in which a certain set of players must win, \ie, meet their objectives.
%Notice that the problem is about the \emph{existence} of such solutions, so it
%is a \emph{cooperative} setting in which we assume all the agents will cooperate
%when presented with an NE. Note that since the payoff of each agent for a
%strategy profile is either $0$ or $1$, we consider threshold tuples $v$ with
%Boolean components. Thus the threshold decision problem is equivalent to
%checking whether there is an NE in which a given set of agents is included in
%the set of winners.

For completeness, we start by referring to the \emph{Constrained NE Problem}
solved by Bouyer \etal~\cite{bouyer2015pure}. This problem asks whether there is an NE in a given game whose payoff profile satisfies both an upper and a lower threshold.\footnote{The original formulation in \cite{bouyer2015pure} was stated in terms of a preference relation $\lesssim_A$ for each agent $A$ as follows: Given two plays $\rho_A^\ell$ and $\rho_A^u$ for each agent~$A$, is there a Nash equilibrium $\bar{\sigma}$ which satisfies $\rho_A^\ell \lesssim_A \out(\bar{\sigma}) \lesssim_A \rho_A^u$ for all $A \in \agt$? Since for $\omega$-regular objectives a threshold play $\rho$ is encoded by the pair $(\occ(\rho),\infi(\rho))$ in \cite{bouyer2015pure}, the way we state the problem is polynomially equivalent.}

\begin{prob}[\bf Constrained NE Existence Problem] Given a game $\cG$, a state $s$ in $\cG$ and 
threshold tuples $\mathbf{v},\mathbf{u} \!\in\! \B^n$, decide whether there
exists an NE $\bar{\sigma}$ such that $\mathbf{v} \leq (\pay(s,\bar{\sigma})) \leq \mathbf{u}$,
where the ordering is defined componentwise.
\end{prob}

Bouyer \etal\ showed that the problem is $\p$-complete for B{\"u}chi objectives,
$\np$-complete for safety, reachability and co-B{\"u}chi objectives, and
$\pspace$-complete for Muller objectives; the complexity class for parity
objectives is $\p^{\np}_\parallel$-complete, where $\p^{\np}_\parallel$ is the class of problems
that can be solved by a deterministic Turing machine in polynomial time with an
access to an oracle for solving $\np$ problems and such that the oracle can be queried only 
once with a set of queries.

We now come to the relevant equilibrium problems that we explore in this paper.
First, we define the social welfare problem as placing a lower bound on the social
welfare, \ie, the number of winning players, of an equilibrium.

\begin{prob}[\bf Social Welfare Decision Problem ($\swdp$)]
Given a game $\cG$, a state $s$ in $\cG$ and a threshold value $v \in \N$, decide whether \emph{there exists} an NE $\bar{\sigma}$ such that $\sw(s,\bar{\sigma}) \geq v$.
\end{prob}

The second problem we consider is the Pareto optimality decision problem for rational synthesis. 

\begin{prob}[\bf Pareto Optimal Decision Problem ($\podp$)]
Given a game $\cG$ and a state $s$ in $\cG$, decide whether \emph{there exists} an NE $\bar{\sigma}$ such that $\pay(s,\bar{\sigma})$ is Pareto optimal.
\end{prob}

Notice that the problems above are about the \emph{existence} of NE satisfying
some conditions. Hence, it is a \emph{cooperative} setting in which we assume that all agents will cooperate when presented with an NE.

\subsection{The Suspect Game}
\label{ubsec:susp}
Many results from Bouyer \etal~\cite{bouyer2015pure} that we use below rely on  a key construction. The idea is based on a correspondence between Nash equilibria in a concurrent game $\cG$ and winning strategies in a two-player zero-sum game $\cH$ derived from $\cG$, called the suspect game. The game $\cH$ is played between \eve\ and \adam. Intuitively, \eve's goal is to prove that the sequence of moves proposed by her results from a Nash equilibrium in $\cG$, while \adam's task is to foil her attempt by exhibiting that some agent has a profitable deviation from the strategy suggested by \eve. Here we recall the basic definitions and results from \cite{bouyer2015pure}.

Given two states $s$ and $s'$ and a move $\bar{a}$ in a concurrent game $\cG$, the set of \emph{suspect agents} for $(s,s')$ and $\bar{a}$ is the set
\[ \susp((s,s'),\bar{a}) = \{ i \in \agt \mid \exists a' \in \avb(s,i).\  
\tr(s,\langle\bar{a}_{-i},a' \rangle) = s' \}. \]
Note that if $\tr(s,\bar{a})=s'$ then $\susp((s,s'),\bar{a}) = \agt$, \ie, every agent is suspect if there is no deviation from the suggested move.
For a play $\rho$ and a strategy profile $\bar{\sigma}$, the set of suspect agents for $\rho$ and $\bar{\sigma}$ is given by the set of suspect agents along each transition of $\rho$:
\[ \susp(\rho,\bar{\sigma}) = \{ i \in \agt \mid \forall j \in \N.\ 
i \in \susp( (\rho[j],\rho[j+1]),\bar{\sigma}(\rho[0,j]) \}. \]
The idea is that agent~$i$ is a suspect for a pair $(s,s')$ and move $\bar{a}$ if she can unilaterally deviate from her action $a_i$ in $\bar{a}$ to trigger the transition $(s,s')$. It follows from the above definitions that agent~$i$ is in $\susp(\rho,\bar{\sigma})$ if, and only if, there is a strategy $\sigma'$ for agent~$i$ such that $\out(\langle\bar{\sigma}_{-i},\sigma' \rangle) = 
\rho$.

For a fixed play $\pi$ in a concurrent game $\cG$, we build the \emph{suspect game} $\cH(\cG,\pi)$, a two-player turn-based zero-sum game between \eve\ and \adam\ as follows. The set of states of $\cH(\cG,\pi)$ is the disjoint union of the states $V_\exists \subseteq \state \times 2^\agt$ owned by \eve, and the set $V_\forall \subset \state \times 2^\agt \times \act^\agt$ owned by \adam. The game proceeds in the following way: from a state $(s,P)$ in $V_\exists$, \eve\ chooses a legal move $\bar{a}$ from $s$ in $\cG$, resulting in the new state $(s,P,\bar{a})$ in $V_\forall$. \adam\ then chooses a move $\bar{a}'$ that will actually apply in $\cG$ leading to a state $s'$ in $\state$; the resulting state in $\cH(\cG,\pi)$ is $(s',P \cap \susp((s,s'),\bar{a}))$. In the special case when the state chosen by \adam\ is such that $s' = \tr(s,\bar{a})$, we say that \adam\ obeys \eve. In this case, the new state is given by $(s',P)$.

We define the two projections $\proj_1$ and $\proj_2$ from $V_\exists$ to $\state$ and $2^\agt$, respectively, by $\proj_1(s,P)=s$ and $\proj_2(s,P)=P$. These projections are extended to plays in $\cH(\cG,\pi)$ in a natural way but only using \eve's states -- for example, $\proj_1((s_0,P_0)(s_0,P_0,\bar{a})(s_1,P_1))\cdots = s_0s_1\cdots$. For any play $\rho$ in $\cH(\cG,\pi)$, $\proj_2(\rho)$, which is a sequence of sets of agents in $\cG$, is non-increasing, and hence its limit $\lambda(\rho)$ is well defined. Note that if $\lambda(\rho) \neq \emptyset$ then $\proj_1(\rho)$ is a play in $\cG$. A play $\rho$ is winning for \eve, if for all $i \in \lambda(\rho)$ the play $\pi$ is as good as or better than $\proj_1(\rho)$ for agent~$i$ in $\cG$, \ie., $\pay_i(\pi) \geq \pay_i(\proj_1(\rho))$. The \emph{winning region} $W(\cG,\pi)$ is the set of states of $\cH(\cG,\pi)$ from which \eve\ has a winning strategy.

The correctness of the suspect game construction is captured by the following theorem from \cite{bouyer2015pure}.

\begin{theorem}
Let $\cG$ be a concurrent game, $s$ a state of $\cG$ and $\pi$ a play in $\cG$. Then the following two conditions are equivalent.
\be
\item There is an NE $\bar{\sigma}$ from $s$ in $\cG$ whose outcome is $\pi$.
\item There is a play $\rho$ from $(s,\agt)$ in $\cH(\cG,\pi)$ satisfying
  \be
  \item \adam\ obeys \eve\ along $\rho$,
  \item $\proj_1(\rho)=\pi$, and
  \item for all $i \in \N$, there is a strategy $\sigma_\exists^i$ for \eve, for which any play in $\rho[0,i]\cdot\out(\rho[i],\sigma_\exists^i)$ is winning for \eve.
  \ee
\ee
\end{theorem}

The following theorem from \cite{bouyer2015pure} asserts that not only is the suspect game construction correct, but it does not result in an exponential blow-up in size as well. Note that the infinite play $\pi$ is specified by the pair of finite sets $(\occ(\pi),\infi(\pi))$.

\begin{theorem}
  Let $\cG$ be a concurrent game and $\pi$ a play in $\cG$. The number of reachable states from $\state \times \agt$ in $\cH(\cG,\pi)$ is polynomial in the size of $\cG$.
\end{theorem}

Note that for two plays $\pi$ and $\pi'$ in a concurrent game $\cG$ the suspect games
$\cH(\cG,\pi)$ and $\cH(\cG,\pi')$ have identical CGSs and differ only in the winning conditions. Thus, the CGS $\cJ(\cG)$ of $\cH(\cG,\pi)$ depends solely on $\cG$ and is polynomial in size. Also, if we denote the set of losing agents in $\pi$ by $\mathrm{Los}(\pi)$, the winning condition for \eve\ in $\cH(\cG,\pi)$ can be stated as follows: for every $i \in \lambda(\rho) \cap \mathrm{Los}(\pi)$, agent~$i$ loses $\proj_1(\rho)$ in $\cG$. Thus, the winning condition depends only on $\mathrm{Los}(\pi)$ and not the exact sequence $\pi$. Henceforth we denote the suspect game by $\cH(\cG,L)$, where $L \subseteq \agt$ and \eve\ wins the play $\rho$ if for every $i \in \lambda(\rho) \cap L$, agent~$i$ loses the play $\proj_1(\rho)$ in $\cG$.

\section{Relevant Equilibria for Omega-regular Concurrent Games}
\label{sec:qual}

\subsection{Reachability Games}
\label{subsec:reach}

\paragraph{\bf Social Welfare Problem}
We begin by showing that $\swdp$, the social welfare decision problem, is
$\np$-complete for reachability objectives. The reduction from $\sat$ is based
on a construction by Bouyer \etal~\cite{bouyer2015pure} for the
constrained NE existence problem for turn-based games.

\begin{figure}[t]
\begin{center}  
\caption{Reduction from $\sat$ to $\swdp$: Reachability}
\vspace{0.5cm}
\label{fig:swdp-reach}
\begin{tikzpicture}
\node[state, initial] (x1) {$ $};
\node[state, above right of=x1] (1x1) {${x_1}$};
\node[state, below right of=x1] (0x1) {${\neg x_1}$};
\node[state, below right of=1x1] (x2) {$ $};
\node[state, above right of=x2] (1x2) {${x_2}$};
\node[state, below right of=x2] (0x2) {${\neg x_2}$};
\node[state, below right of=1x2] (x3) {$ $};
\node[state, draw=none] (qdots) [above right=of x3] {$\cdots$}; 
\node[state, draw=none] (rdots) [below right=of x3] {$\cdots$}; 
\node[state, below right=of qdots] (xn) {$ $};
\node[state, above right of=xn] (1xn) {${x_n}$};
\node[state, below right of=xn] (0xn) {${\neg x_n}$};
\node[state, below right of=1xn] (C1) {$ $};

\draw 
 (x1) edge (1x1)
 (x1) edge (0x1)
 (1x1) edge (x2)
 (0x1) edge (x2)
 (x2) edge (1x2)
 (x2) edge (0x2)
 (1x2) edge (x3)
 (0x2) edge (x3)
 (x3) edge (qdots)
 (x3) edge (rdots)
 (qdots) edge (xn)
 (rdots) edge (xn)
 (xn) edge (1xn)
 (xn) edge (0xn)
 (1xn) edge (C1)
 (0xn) edge (C1)
 (C1) edge[loop right] (C1);
\end{tikzpicture}
\end{center}        
\end{figure}
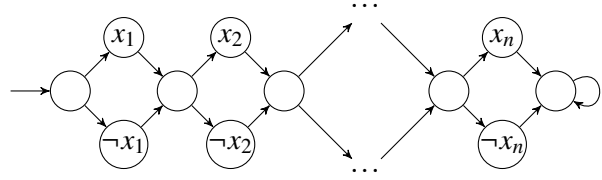

\begin{lem}
\label{lem:swdp-reach}
$\swdp$ for reachability objectives is $\np$-complete.
\end{lem}

\begin{proof}
We reduce $\sat$ to $\swdp$ with reachability objectives. Let $\phi = C_1 \wedge
\ldots \wedge C_m$ be an instance of $\sat$ with $C_i = \ell_{i,1} \vee
\ell_{i,2} \vee \ell_{i,3}$ over a set of variables $\{x_1,\ldots,x_n\}$. Then
construct the $(m+1)$-player turn-based reachability game as shown in
Fig.~\ref{fig:swdp-reach}. In this game Player~$0$ owns all the states, shown as
circles in the figure. The target set for Player~$0$ includes all the states,
\ie, Player~$0$ always wins. The threshold value $v$ for the social welfare
function in the $\swdp$ instance is set to $m+1$, the number of players.
Clearly, the game has an NE with a payoff of $1$ for each player if, and only if,
$\phi$ is satisfiable. If $\tau$ is a satisfying valuation for $\phi$, the
strategy for Player~$0$ is simple: between $x_k$ and $\neg x_k$ choose $x_k$ if $\tau(x_k)=1$ and $\neg x_k$ otherwise. Conversely, if all the players win, a satisfying valuation $\tau$ can
be similarly obtained from Player~$0$'s strategy, by setting $\tau(x_k)$ to $1$ if Player~$0$ chooses to move to $x_k$ and  to $0$ otherwise.

We show that $\swdp$ for reachability objectives is in
$\np$ by mirroring the Algorithm in Section~5.1.2 in \cite{bouyer2015pure}
with a small modification for adapting the solution for the constrained NE existence problem to
$\swdp$. The algorithm shown below makes essential use of the suspect game $\cH(\cG,L)$ and its reduction
to the safety game $\cJ(\cG)$ with safety objective $\Omega_L$. See \cite{bouyer2015pure} for the technical details.

\be
\item Given a value $v \in [0,n]$, first guess a lasso-shaped play $\rho = \alpha_1 \cdot \alpha_2^\omega$ where $|\alpha_i|^2 \leq 2|\state|^2$
in $\cJ(\cG)$ such that \adam\ obeys \eve\ along $\rho$, and the play $\pi = \proj_1(\rho)$ in $\cG$ satisfies the constraint that at least $v$ players are winning in it. This condition on the number of winning players is the only change from Section~5.1.2 in \cite{bouyer2015pure}.
\item Then compute the set $W(\cG, \mathrm{Los}(\pi))$ of the winning states for \eve\ in the suspect game $\cH(\cG,\mathrm{Los}(\pi))$, where $\mathrm{Los}(\pi)$ is the set of losing players along $\pi$.
\item Finally, check that $\rho$ always stays in $W(\cG, \mathrm{Los}(\pi))$.
\ee
We refer to \cite{bouyer2015pure} for the proof that this nondeterministic algorithm runs in polynomial time.

\end{proof}

\paragraph{\bf Pareto Optimal Decision Problem}

It is clear that the construction in the proof of Lemma~\ref{lem:swdp-reach}
yields an NE that is Pareto-optimal if, and only if, the formula $\phi$ is
satisfiable. Hence the Pareto Optimal Decision Problem for reachability objectives
is $\np$-hard using the same reduction as in the lemma.

\begin{lem}
\label{lem:podp-reach}
$\podp$ for reachability objectives is $\np$-hard.
\end{lem}

% Upper bound for Pareto optimality (reachability)
For the upper bound, we show that $\podp$ for reachability objectives is in
the class $\p^\np = \mathsf{\Delta}^p_2$ in the polynomial hierarchy as follows.
First, use binary search for the threshold value $v$ in the range $[0,n]$ and the $\np$ oracle for deciding $\swdp$ in the proof of Lemma~\ref{lem:swdp-buchi} above to determine the maximum value $m$ of $v$ for which the procedure returns yes. Then, to check if there is a run $\rho$ where more than $m$ players are winners, search for any strongly connected component (SCC) $C$ reachable from the initial state $s$ in the underlying CGS that satisfies the following condition: there are more than $m$ sets in $F_1,\ldots,F_n$ with a nonempty intersection with $C$. Since this can be done in polynomial time using Tarjan's algorithm for finding all SCCs in a directed graph~\cite{tarjan1972depth}, the entire procedure runs in $\p^\np$ time. However, we leave the question of whether $\podp$ for reachability objectives is $\p^\np$-hard open.

\subsection{Safety Games}
\label{subsec:safe}

\paragraph{\bf Social Welfare Problem}

We show that $\swdp$ is $\np$-complete for safety objectives by reduction from
$\sat$. We use a modification of a construction by Bouyer
\etal~\cite{bouyer2015pure} for the value problem for ordered B{\"u}chi
objectives with the counting preorder.

\begin{lem}
\label{lem:swdp-safe}
$\swdp$ for safety objectives is $\np$-complete.
\end{lem}

\begin{proof}
We reduce $\sat$ to $\swdp$ with safety objectives for a turn-based game.
Consider an instance $\phi = C_1 \wedge \ldots \wedge C_m$ of $\sat$ with $C_j =
\ell_{j,1} \vee \ell_{j,2} \vee \ell_{j,3}$ over a set of variables
$\{x_1,\ldots,x_n\}$. We associate a $(2n+1)$-player turn-based
game $\cG_\phi$ with $\phi$. The set of states is given by the union of
\begin{align*}
&V_0 =\{s\} \text{ where } s \text{ is the initial state in the $\swdp$ instance,}\\
&V_k = \{x_k,\neg x_k\} \text{ for each } 1 \leq k \leq n, \text{ and }\\
&V_{n+j} =\{\ell_{j,1}, \ell_{j,2},\ell_{j,3}\} \text{ for each } 1 \leq j \leq m.
\end{align*}
We add a transition from each state in $V_i$ to each state in $V_{i+1}$ for $0
\leq i \leq n + m$, assuming $V_{n+m+1} = V_0$. The game has $2n+1$ players
$P_0,\ldots,P_{2n}$. $P_0$ owns all the states and her safety objective is given
by the unsafe set $F_0 = \emptyset$, \ie, she always wins. This implies that all
strategy profiles are NE, as the other players have no choice at any state. For
$1 \leq k \leq n$, $P_{2k-1}$ and $P_{2k}$ have safety objectives given by the
unsafe states
\begin{align*}
&F_{2k-1} = \{x_k\} \cup \{\ell_{j,p} \mid \ell_{j,p} \text{ is the literal } x_k\} \\
&F_{2k} = \{\neg x_k\} \cup \{\ell_{j,p} \mid \ell_{j,p} \text{ is the literal } \neg x_k\} 
\end{align*}
At least $n$ of these sets $F_1,\ldots,F_{2n}$ will be visited along any
infinite play and thus at least $n$ of these $2n$ players $P_1,\ldots, P_{2n}$
will always lose. We show that $\phi$ is satisfiable if, and only if, there exists
an NE for which at most (and hence exactly) $n$ of these $2n$ sets
$F_1,\ldots,F_{2n}$ are visited, \ie, at least $n+1$ players win. In other
words, $\phi$ is satisfiable if, and only if, there is an NE $\sigma$ in
$\cG_\phi$ with $\sw(\sigma) \geq n+1$.

Assume $\phi$ is satisfiable and let $\tau$ be a satisfying valuation. The
strategy $\sigma_0$ for $P_0$ simply follows $\tau$, \ie, for states in
$V_{k-1}$ for $1 \leq k \leq n$, the strategy chooses $x_k$ if $\tau(x_k) =
\true$ and it chooses $\neg x_k$ otherwise. From a state in $V_{n+k-1}$ for $1
\leq k \leq m$, it chooses one of the $\ell_{j,p}$ that evaluates to $\true$
under $\tau$, say the one with the least index. As a result, the number of sets in
$F_1, \ldots, F_n$ that are visited at least once is $n$. Since $P_0$ owns all
the states and wins using $\sigma_0$ and the strategies of all the other players
are empty, we have an NE $\sigma$ where $n+1$ players win.

Conversely, pick a play in $\cG$ resulting from a strategy $\sigma_0$ of $P_0$
such that at most (hence exactly) $n$ of the sets $F_1,\ldots,F_{2n}$ are
visited at least once. In particular, this play never visits one of $x_k$ and
$\neg x_k$ for any $1 \leq k \leq n$. Clearly, such a strategy is part of an NE
$\sigma$ (since all strategy profiles are) in which the number of winning
players is $n+1$. One can define a truth valuation $\tau$ over
$\{x_1,\ldots,x_n\}$ from the play -- simply set $\tau(x_k)$ to $\true$ if $x_k$
is visited at least once and to $\false$ otherwise. Also, any state of $V_{n+j}$
with $1 \leq j \leq m$ that is visited at least once must correspond to a
literal that is assigned the value $\true$ by $\tau$, otherwise there would be
more than $n$ states among $F_1,\ldots,F_n$ visited at least once. Hence, each
clause of $\phi$ evaluates to $\true$ under $\tau$, and therefore $\phi$ is
satisfiable.

To show that $\swdp$ for safety objectives is in $\np$, we follow the Algorithm in Section~5.2.3 in \cite{bouyer2015pure} with a small modification, just as in the reachability case in 
Section~\ref{subsec:reach}. The algorithm reduces the suspect game $\cH(\cG,L)$
to the safety game $\cJ(\cG)$ with safety objective $\Omega_L$. See \cite{bouyer2015pure} 
for the technical details.

\be
\item Given a value $v \in [0,n]$, first guess a lasso-shaped play $\rho = \alpha_1 \cdot \alpha_2^\omega$ where $|\alpha_i|^2 \leq |\state|^2$ in $\cJ(\cG)$ such that \adam\ obeys \eve\ along $\rho$, and the play $\pi = \proj_1(\rho)$ in $\cG$ satisfies the constraint that at least $v$ players are winning in it.
\item Then check that any deviation by \adam\ along $\rho$, say at position $i$ leads to a state from which \eve\ has a strategy $\sigma^i$ that ensures that any play in $\rho[0,i] \cdot \out(\sigma^i)$ is winning.
\ee

We refer to \cite{bouyer2015pure} for the details of step 2 above and the proof that this nondeterministic algorithm runs in polynomial time.

\end{proof}

\paragraph{\bf Pareto Optimal Decision Problem}

It is clear that the construction in the proof of Lemma~\ref{lem:swdp-safe}
yields an NE that is Pareto-optimal if, and only if, the formula $\phi$ is
satisfiable, just as in the case of reachability. 

\begin{lem}
\label{lem:podp-safe}
$\podp$ for safety objectives is $\np$-hard.
\end{lem}

% Upper bound for Pareto optimality (safety)
For the upper bound, we show that $\podp$ for safety objectives is in
the class $\p^\np = \mathsf{\Delta}^p_2$ in the polynomial hierarchy using the same procedure as in the reachability case, except now we use the $\np$ oracle for deciding $\swdp$ for safety objectives. As in the case of reachability, we leave the question of whether $\podp$ for safety objectives is $\p^\np$-hard open.

\subsection{B{\"u}chi Games}
\label{subsec:buchi}

\paragraph{\bf Social Welfare Problem}
We show that the social welfare problem for B{\"u}chi objectives can be solved in polynomial time by giving a polynomial time algorithm that invokes the procedure for Constrained NE Existence Problem from \cite{bouyer2015pure}. A polynomial-time algorithm for the latter was presented in \cite{bouyer2015pure}.

\begin{lem}
\label{lem:swdp-buchi}
$\swdp$ for B{\"u}chi objectives is in $\p$.
\end{lem}

\begin{proof}
The following is a polynomial-time algorithm for deciding whether there is a Nash equilibrium with $v$
or more winners in a given concurrent game with B{\"u}chi objectives, starting from a state $s$.
\be
\item Find all the reachable SCCs in the underlying digraph of $\cG$ using Tarjan's algorithm~\cite{tarjan1972depth}. Call the number of agents that meet their B{\"u}chi objectives in the SCC $C$ the \emph{rank of} $C$. Agent~$i$ meets her objective in the SCC $C$ if $F_i \cap C \neq \emptyset$. 
\item Sort the SCCs in non-increasing order according to their rank.
\item For each rank $r$ starting from the highest down to $v$, check whether there is an NE with $r$ winners using the
algorithm for constrained NE Existence from \cite{bouyer2015pure} by setting both the lower and upper thresholds $\mathbf{v}$ and $\mathbf{u}$ to the winner profile of each SCC $C$ of rank $r$ one by one. If such an NE exists then return `yes', else return `no'.
\ee
\end{proof}

The correctness of the algorithm depends on Proposition~\ref{prop:lasso} on the lasso characterization of
NEs in concurrent games.

\paragraph{\bf Pareto Optimal Decision Problem}

\begin{lem}
\label{lem:podp-buchi}
$\podp$ for B{\"u}chi objectives is in $\p$.
\end{lem}

\begin{proof}
A small modification in the algorithm in the proof of Lemma~\ref{lem:swdp-buchi} gives a polynomial-time algorithm for deciding whether a Pareto-optimal NE exists. After sorting the SCCs in non-increasing order according to their rank, starting from the highest rank we check if any of the SCCs with the given rank is an NE. If a rank r is found for which all SCCs are non-NEs, then return `no'. Otherwise, if an SCC corresponding to an NE is found, then return `yes'.
\end{proof}

\subsection{CoB{\"u}chi Games}
\label{subsec:cobuchi}

\paragraph{\bf Social Welfare Problem}

We show that $\swdp$ is $\np$-complete for coB{\"u}chi objectives by reduction
from $\sat$. The reduction is the same as in Section~\ref{subsec:safe} for
safety objectives, with the unsafe states now playing the role of coB{\"u}chi objectives.

\begin{lem}
\label{lem:swdp-cobuchi}
$\swdp$ for coB{\"u}chi objectives is $\np$-complete.
\end{lem}

\begin{proof}
We reduce $\sat$ to $\swdp$ with coB{\"u}chi objectives for a turn-based game
using the same construction as in the proof of Lemma~\ref{lem:swdp-safe}, with the unsafe states $F_1,\ldots,F_n$ being now designated as the coB{\"u}chi objectives.

At least $n$ of these sets $F_1,\ldots,F_{2n}$ will be visited infinitely often along any infinite play and thus at least $n$ of these $2n$ players $P_1,\ldots, P_{2n}$ will always lose. We show that $\phi$ is satisfiable if, and only if,
there exists an NE for which at most (and hence exactly) $n$ of these $2n$ sets
$F_1,\ldots,F_{2n}$ are visited infinitely often, \ie, at least $n+1$ players
win. In other words, $\phi$ is satisfiable if, and only if, there is an NE
$\sigma$ in $\cG_\phi$ with $\sw(\sigma) \geq n+1$.

Assume $\phi$ is satisfiable and let $\tau$ be a satisfying valuation. The
strategy $\sigma_0$ for $P_0$ simply follows $\tau$, \ie, for states in
$V_{k-1}$ for $1 \leq k \leq n$, the strategy chooses $x_k$ if $\tau(x_k) =
\true$ and it chooses $\neg x_k$ otherwise. From a state in $V_{n+k-1}$ for $1
\leq k \leq m$, it chooses one of the $\ell_{j,p}$ that evaluates to $\true$
under $\tau$, say the one with the least index. This way the number of sets in
$F_1, \ldots, F_n$ that are visited infinitely often is $n$ and the other sets
are not visited at all. Since $P_0$ owns all the states and wins using
$\sigma_0$ and the strategies of all the other players are empty, we have an NE
$\sigma$ where $n+1$ players win.

Conversely, pick a play in $\cG$ resulting from a strategy $\sigma_0$ of $P_0$
such that at most (hence exactly) $n$ of the sets $F_1,\ldots,F_{2n}$ are
visited infinitely often. In particular, this play always visits one of $x_k$ and
$\neg x_k$ finitely often for any $1 \leq k \leq n$. Clearly, such a strategy is
part of an NE $\sigma$ (since all strategy profiles are) in which the number of
winning players is $n+1$. One can define a truth valuation $\tau$ over
$\{x_1,\ldots,x_n\}$ from the play -- simply set $\tau(x_k)$ to $\true$ if $x_k$
is visited infinitely often and to $\false$ otherwise. Also, any state of
$V_{n+j}$ with $1 \leq j \leq m$ that is visited infinitely often must
correspond to a literal that is assigned the value $\true$ by $\tau$, otherwise
there would be more than $n$ states among $F_1,\ldots,F_n$ visited infinitely
often. Hence each clause of $\phi$ evaluates to $\true$ under $\tau$, and hence
$\phi$ is satisfiable.

To show that $\swdp$ for coB{\"u}chi objectives is in
$\np$, we follow the Algorithm in Section~5.4.3 in \cite{bouyer2015pure}
with a small modification as in the reachability case above. The algorithm also uses the suspect game $\cH(\cG,L)$ and its reduction to the safety game $\cJ(\cG)$ with safety objective $\Omega_L$. See \cite{bouyer2015pure} for the technical details.

\be
\item Given a value $v \in [0,n]$, first guess a lasso-shaped play $\rho = \alpha_1 \cdot \alpha_2^\omega$ where $|\alpha_i|^2 \leq |\state|^2$ in $\cJ(\cG)$ such that \adam\ obeys \eve\ along $\rho$, and the play $\pi = \proj_1(\rho)$ in $\cG$ satisfies the constraint that at least $v$ players are winning in it.
\item Then check that any deviation by \adam\ along $\rho$, say at position $i$ leads to a state from which \eve\ has a strategy $\sigma^i$ that ensures that any play in $\rho[0,i] \cdot \out(\sigma^i)$ is winning.
\ee

We refer to \cite{bouyer2015pure} for the details of step 2 above and the proof that this nondeterministic algorithm runs in polynomial time.

\end{proof}

\paragraph{\bf Pareto Optimal Decision Problem}

As in the case of the other objectives considered above, the construction in the
proof of Lemma~\ref{lem:swdp-cobuchi} yields an NE which is Pareto-optimal if and
only if the formula $\phi$ is satisfiable.

\begin{lem}
\label{lem:podp-cobuchi}
$\podp$ for coB{\"u}chi objectives is $\np$-hard.
\end{lem}

% Upper bound for PODP for coBuchi
For the upper bound, we show that $\podp$ for coB{\"u}chi objectives is in
the class $\p^\np = \mathsf{\Delta}^p_2$ in the polynomial hierarchy using the same procedure as in the safety case, except now we use the $\np$ oracle for deciding $\swdp$ for coB{\"u}chi objectives. As in the case of safety, we leave the question whether $\podp$ for coB{\"u}chi objectives is $\p^\np$-hard open.

\subsection{Parity Games}
\label{subsec:parity}

\paragraph{\bf Social Welfare Problem}

Bouyer \etal~\cite{bouyer2015pure} showed that the Constrained NE Existence Problem is 
$\p^{\np}_\parallel$-complete for parity objectives by reduction from $\osat$.
Intuitively, $\p^{\np}_\parallel$ is the class of all
languages accepted by some deterministic polynomial time Turing machine $M$
using an oracle for solving $\np$ problems, such that on any input the machine
$M$ builds a set of queries to the oracle before making the queries just once. For
a formal definition see \cite{wagner1990bounded}. In the $\osat$ problem, given
a finite set of instances of $\sat$, the goal is to decide whether the number of
satisfiable instances is even. The problem is known to be
$\p^{\np}_\parallel$-complete~\cite{gottlob1995np}.

We show the same upper bound for $\swdp$ for parity objectives, \ie, it is in $\p^{\np}_\parallel$. The proof is essentially the same algorithm as in \cite{bouyer2015pure} -- see Section~5.6.2 in the paper for Rabin objectives.
We first translate the parity objectives to corresponding Rabin ones with half as many pairs as the number of priorities and then apply the algorithm in \cite{bouyer2015pure}.
The only modification in the algorithm is identical to the ones for the reachability, safety and coB{\"u}chi objectives, namely step 1, where we check that the play $\pi = \proj_1(\rho)$ in $\cG$ satisfies the constraint that at least $v$ players are winning in it, where $v$ is the threshold input to $\swdp$.

\begin{lem}
\label{lem:swdp-parity}
$\swdp$ for parity objectives is in $\p^{\np}_\parallel$.
\end{lem}

However, for the lower bound we leave the question whether $\swdp$ for parity objectives is $\p^{\np}_\parallel$-hard as open. However, we can show the weaker result that the problem is $\np$-hard just by coding the coB{\"u}chi condition by a parity condition with two colours.

\paragraph{\bf Pareto Optimal Decision Problem}
We show that $\podp$ for parity objectives is in
the class $\p^{\p^{\np}} = \p^{\np}$ in the polynomial hierarchy using the same procedure as in the coB{\"u}chi case, except now we use the $\p^{\np}_\parallel$ oracle for deciding $\swdp$ for parity objectives. Again, we leave the question of whether $\podp$ for parity objectives is $\p^\np$-hard open.

\subsection{Muller Games}
\label{subsec:cmuller}

\paragraph{\bf Social Welfare Problem}

We show that $\swdp$ is $\pspace$-complete for Muller objectives by a reduction from
$\tqbf$. 

\begin{lem}
\label{lem:swdp-muller}
$\swdp$ for Muller objectives is $\pspace$-complete.
\end{lem}

\begin{proof}
We reduce $\tqbf$ to $\swdp$ with Muller objectives for a three-player turn-based game by reusing the proof
of $\pspace$-hardness of deciding the winner of a zero-sum two-player Muller game by Hunter and Dawar~\cite{hunter2005complexity}. First, note that the complementary objective of a zero-sum Muller game is also a Muller game, simply by changing the colouring function appropriately. Then given a $\tqbf$ formula $\varphi$
we take the two-player Muller game $\cG_\varphi$ in \cite{hunter2005complexity} and simply add one more player
(Player~$2$) who has the same Muller objective as Player~$0$ but controls no vertex. In other words, for all
plays $\rho$, Player~$2$ wins $\rho$ iff Player~$0$ wins $\rho$ iff Player~$1$ loses $\rho$.
Thus, setting the threshold value $v$ in $\swdp$ to $2$ and using the construction in \cite{hunter2005complexity} with the above modification reduces $\tqbf$ to $\swdp$ with Muller objectives for three players.

The proof of membership of $\swdp$ for Muller objectives in $\pspace$ is as
follows. Given a value $v \in [0,n]$, first guess a set $W$ of $v$ winning
players. This can clearly be done in polynomial space. Then use the procedure
for checking membership in $\pspace$ for the corresponding Constrained NE
Problem from \cite{bouyer2015pure} (see the third paragraph on page 25 of the
paper) with appropriate lower and upper threshold tuples of bits, $\mathbf{v}$ 
and $\mathbf{u}$, respectively. Here, $\mathbf{v}$ contains $1$'s only for the $v$
winners and $0$'s elsewhere, and $\mathbf{u}$ contains all $1$'s.
\end{proof}

\paragraph{\bf Pareto Optimal Decision Problem}
As in the previous cases, the construction in the proof of
Lemma~\ref{lem:swdp-muller} produces an NE that is Pareto-optimal for Muller
objectives if, and only if, the formula $\phi$ is satisfiable. Membership in the class $\pspace$ follows using the same algorithm for all earlier cases, except using the $\pspace$ oracle to decide $\swdp$ for Muller objectives.

\begin{lem}
\label{lem:podp-muller}
$\podp$ for Muller objectives is $\pspace$-complete.
\end{lem}

\section{Conclusion}
In this work, we have extended the complexity results for rational synthesis problems for concurrent games to the case of relevant equilibria. We restrict ourselves to pure strategy
Nash equilibria satisfying  a social welfare or a Pareto optimality condition.

This work can be extended in many possible directions. One can consider solution concepts other than Nash equilibria, such as subgame perfect equilibria and admissible strategy profiles. It will also be fruitful to extend these results to the case of quantitative games such as those with mean payoff and discounted-sum objectives. Finally, the case of mixed rather than pure strategy equilibria would be an interesting extension. 

%\nocite{*}
\bibliographystyle{eptcs}
\bibliography{rational.bib}    

\begin{thebibliography}{10}
\providecommand{\bibitemdeclare}[2]{}
\providecommand{\surnamestart}{}
\providecommand{\surnameend}{}
\providecommand{\urlprefix}{Available at }
\providecommand{\url}[1]{\texttt{#1}}
\providecommand{\href}[2]{\texttt{#2}}
\providecommand{\urlalt}[2]{\href{#1}{#2}}
\providecommand{\doi}[1]{doi:\urlalt{https://doi.org/#1}{#1}}
\providecommand{\eprint}[1]{arXiv:\urlalt{https://arxiv.org/abs/#1}{#1}}
\providecommand{\bibinfo}[2]{#2}

\bibitemdeclare{inproceedings}{alur1997alternating}
\bibitem{alur1997alternating}
\bibinfo{author}{Rajeev \surnamestart Alur\surnameend},
  \bibinfo{author}{Thomas~A \surnamestart Henzinger\surnameend} \&
  \bibinfo{author}{Orna \surnamestart Kupferman\surnameend}
  (\bibinfo{year}{1997}): \emph{\bibinfo{title}{Alternating-time temporal
  logic}}.
\newblock In: {\slshape \bibinfo{booktitle}{International Symposium on
  Compositionality}}, \bibinfo{organization}{Springer}, pp.
  \bibinfo{pages}{23--60}, \doi{10.1007/3-540-49213-5_2}.

\bibitemdeclare{incollection}{BloemCJ18}
\bibitem{BloemCJ18}
\bibinfo{author}{Roderick \surnamestart Bloem\surnameend},
  \bibinfo{author}{Krishnendu \surnamestart Chatterjee\surnameend} \&
  \bibinfo{author}{Barbara \surnamestart Jobstmann\surnameend}
  (\bibinfo{year}{2018}): \emph{\bibinfo{title}{Graph Games and Reactive
  Synthesis}}.
\newblock In: {\slshape \bibinfo{booktitle}{Handbook of Model Checking}},
  \bibinfo{publisher}{Springer}, pp. \bibinfo{pages}{921--962},
  \doi{10.1007/978-3-319-10575-8_27}.

\bibitemdeclare{article}{bouyer2015pure}
\bibitem{bouyer2015pure}
\bibinfo{author}{Patricia \surnamestart Bouyer\surnameend},
  \bibinfo{author}{Romain \surnamestart Brenguier\surnameend},
  \bibinfo{author}{Nicolas \surnamestart Markey\surnameend} \&
  \bibinfo{author}{Michael \surnamestart Ummels\surnameend}
  (\bibinfo{year}{2015}): \emph{\bibinfo{title}{Pure {N}ash equilibria in
  concurrent deterministic games}}.
\newblock {\slshape \bibinfo{journal}{Logical Methods in Computer Science}}
  \bibinfo{volume}{11}, \doi{10.2168/LMCS-11(2:9)2015}.

\bibitemdeclare{article}{brihaye2021relevant}
\bibitem{brihaye2021relevant}
\bibinfo{author}{Thomas \surnamestart Brihaye\surnameend},
  \bibinfo{author}{V{\'e}ronique \surnamestart Bruy{\`e}re\surnameend},
  \bibinfo{author}{Aline \surnamestart Goeminne\surnameend} \&
  \bibinfo{author}{Nathan \surnamestart Thomasset\surnameend}
  (\bibinfo{year}{2021}): \emph{\bibinfo{title}{On relevant equilibria in
  reachability games}}.
\newblock {\slshape \bibinfo{journal}{Journal of Computer and System Sciences}}
  \bibinfo{volume}{119}, pp. \bibinfo{pages}{211--230},
  \doi{10.1016/j.jcss.2021.02.009}.

\bibitemdeclare{inproceedings}{condurache2016complexity}
\bibitem{condurache2016complexity}
\bibinfo{author}{Rodica \surnamestart Condurache\surnameend},
  \bibinfo{author}{Emmanuel \surnamestart Filiot\surnameend},
  \bibinfo{author}{Raffaella \surnamestart Gentilini\surnameend} \&
  \bibinfo{author}{Jean-Francois \surnamestart Raskin\surnameend}
  (\bibinfo{year}{2016}): \emph{\bibinfo{title}{The complexity of rational
  synthesis}}.
\newblock In: {\slshape \bibinfo{booktitle}{43rd International Colloquium on
  Automata, Languages, and Programming (ICALP 2016)}}, pp.
  \bibinfo{pages}{1--15}, \doi{10.4230/LIPIcs.ICALP.2016.121}.

\bibitemdeclare{incollection}{finkbeiner2016synthesis}
\bibitem{finkbeiner2016synthesis}
\bibinfo{author}{Bernd \surnamestart Finkbeiner\surnameend}
  (\bibinfo{year}{2016}): \emph{\bibinfo{title}{Synthesis of reactive
  systems}}.
\newblock In: {\slshape \bibinfo{booktitle}{Dependable Software Systems
  Engineering}}, \bibinfo{publisher}{IOS Press}, pp. \bibinfo{pages}{72--98},
  \doi{10.3233/978-1-61499-627-9-72}.

\bibitemdeclare{inproceedings}{fisman2010rational}
\bibitem{fisman2010rational}
\bibinfo{author}{Dana \surnamestart Fisman\surnameend}, \bibinfo{author}{Orna
  \surnamestart Kupferman\surnameend} \& \bibinfo{author}{Yoad \surnamestart
  Lustig\surnameend} (\bibinfo{year}{2010}): \emph{\bibinfo{title}{Rational
  synthesis}}.
\newblock In: {\slshape \bibinfo{booktitle}{Tools and Algorithms for the
  Construction and Analysis of Systems: {TACAS} 2010}}, {\slshape
  \bibinfo{series}{LNCS}} \bibinfo{volume}{6015},
  \bibinfo{publisher}{Springer}, pp. \bibinfo{pages}{190--204},
  \doi{10.1007/978-3-642-12002-2_16}.

\bibitemdeclare{article}{gottlob1995np}
\bibitem{gottlob1995np}
\bibinfo{author}{Georg \surnamestart Gottlob\surnameend}
  (\bibinfo{year}{1995}): \emph{\bibinfo{title}{{NP} trees and {C}arnap's modal
  logic}}.
\newblock {\slshape \bibinfo{journal}{Journal of the ACM (JACM)}}
  \bibinfo{volume}{42}(\bibinfo{number}{2}), pp. \bibinfo{pages}{421--457},
  \doi{10.1145/201019.201031}.

\bibitemdeclare{article}{gutierrez2023complexity}
\bibitem{gutierrez2023complexity}
\bibinfo{author}{Julian \surnamestart Gutierrez\surnameend},
  \bibinfo{author}{Muhammad \surnamestart Najib\surnameend},
  \bibinfo{author}{Giuseppe \surnamestart Perelli\surnameend} \&
  \bibinfo{author}{Michael \surnamestart Wooldridge\surnameend}
  (\bibinfo{year}{2023}): \emph{\bibinfo{title}{On the complexity of rational
  verification}}.
\newblock {\slshape \bibinfo{journal}{Annals of Mathematics and Artificial
  Intelligence}} \bibinfo{volume}{91}(\bibinfo{number}{4}), pp.
  \bibinfo{pages}{409--430}, \doi{10.1007/s10472-022-09804-3}.

\bibitemdeclare{inproceedings}{hunter2005complexity}
\bibitem{hunter2005complexity}
\bibinfo{author}{Paul \surnamestart Hunter\surnameend} \& \bibinfo{author}{Anuj
  \surnamestart Dawar\surnameend} (\bibinfo{year}{2005}):
  \emph{\bibinfo{title}{Complexity bounds for regular games}}.
\newblock In: {\slshape \bibinfo{booktitle}{International Symposium on
  Mathematical Foundations of Computer Science}},
  \bibinfo{organization}{Springer}, pp. \bibinfo{pages}{495--506},
  \doi{10.1007/11549345_43}.

\bibitemdeclare{article}{kupferman2016synthesis}
\bibitem{kupferman2016synthesis}
\bibinfo{author}{Orna \surnamestart Kupferman\surnameend},
  \bibinfo{author}{Giuseppe \surnamestart Perelli\surnameend} \&
  \bibinfo{author}{Moshe~Y \surnamestart Vardi\surnameend}
  (\bibinfo{year}{2016}): \emph{\bibinfo{title}{Synthesis with rational
  environments}}.
\newblock {\slshape \bibinfo{journal}{Annals of Mathematics and Artificial
  Intelligence}} \bibinfo{volume}{78}(\bibinfo{number}{1}), pp.
  \bibinfo{pages}{3--20}, \doi{10.1007/s10472-016-9508-8}.

\bibitemdeclare{article}{nash1950equilibrium}
\bibitem{nash1950equilibrium}
\bibinfo{author}{John~F \surnamestart Nash~Jr\surnameend}
  (\bibinfo{year}{1950}): \emph{\bibinfo{title}{Equilibrium points in n-person
  games}}.
\newblock {\slshape \bibinfo{journal}{Proceedings of the National Academy of
  Sciences}} \bibinfo{volume}{36}(\bibinfo{number}{1}), pp.
  \bibinfo{pages}{48--49}, \doi{10.1073/pnas.36.1.48}.

\bibitemdeclare{article}{tarjan1972depth}
\bibitem{tarjan1972depth}
\bibinfo{author}{Robert \surnamestart Tarjan\surnameend}
  (\bibinfo{year}{1972}): \emph{\bibinfo{title}{Depth-first search and linear
  graph algorithms}}.
\newblock {\slshape \bibinfo{journal}{SIAM Journal on Computing}}
  \bibinfo{volume}{1}(\bibinfo{number}{2}), pp. \bibinfo{pages}{146--160},
  \doi{10.1137/0201010}.

\bibitemdeclare{article}{wagner1990bounded}
\bibitem{wagner1990bounded}
\bibinfo{author}{Klaus~W \surnamestart Wagner\surnameend}
  (\bibinfo{year}{1990}): \emph{\bibinfo{title}{Bounded query classes}}.
\newblock {\slshape \bibinfo{journal}{SIAM Journal on Computing}}
  \bibinfo{volume}{19}(\bibinfo{number}{5}), pp. \bibinfo{pages}{833--846},
  \doi{10.1137/0219058}.

\end{thebibliography}
\end{document}